\documentclass[letter]{article}

\usepackage[T1]{fontenc}
\usepackage[utf8x]{inputenc}
\usepackage[english]{babel}
\usepackage{natbib}
\usepackage{color,placeins}
\usepackage{graphicx,amsmath}
\usepackage{graphics} 
\usepackage{epsfig,epstopdf} 
\usepackage{mathptmx} 
\usepackage{times} 
\usepackage{amsmath} 
\usepackage{amssymb}  
\usepackage{amsfonts}
\usepackage{amsmath}
\usepackage{amssymb}
\usepackage{psfrag}
\usepackage{float}

\newtheorem{lemma}{Lemma}[section]
\newtheorem{theorem}{Theorem}[section]
\newtheorem{remark}{Remark}[section]

\def\bi{\begin{itemize}}
\def\ei{\end{itemize}}
\def\bn{\begin{enumerate}}
\def\en{\end{enumerate}}
\def\bq{\begin{eqnarray}}
\def\eq{\end{eqnarray}}
\def\be{\begin{equation}}
\def\ee{\end{equation}}
\def\bea{\begin{eqnarray}}
\def\eea{\end{eqnarray}}
\def\beann{\begin{eqnarray*}}
\def\eeann{\end{eqnarray*}}
\def\bsea{\begin{subeqnarray}}
\def\esea{\end{subeqnarray}}
\def\bmat{\left[ \begin{array}}
\def\emat{\end{array} \right]}
\def\proof{\noindent{\bf{\em Proof:}\ \ }}
\def\QED{\mbox{\rule[0pt]{1.5ex}{1.5ex}}}
\def\endproof{\hspace*{\fill}~\QED\par\endtrivlist\unskip}

\def\s{\star}
\newfont{\BB}{msbm10}
\newfont{\bb}{msbm8}

\def\I{\mbox{\BB I}}

\newcommand{\bmx}{\begin{matrix}}
\newcommand{\emx}{\end{matrix}}
\newcommand{\ba}{\begin{array}}
\newcommand{\ea}{\end{array}}

\def\nn{\nonumber}
\def\bq{\begin{eqnarray}}
\def\eq{\end{eqnarray}}

\def\bsmat{\left[ \begin{smallmatrix}}
\def\esmat{\end{smallmatrix} \right]}

\usepackage[colorlinks=true, allcolors=blue]{hyperref}

\usepackage[sfdefault,lf]{carlito}
\usepackage[parfill]{parskip}

\usepackage{fancyhdr}
\usepackage{natbib}
\usepackage{authblk}
\usepackage[a4paper,top=3cm,bottom=2cm,left=3cm,right=3cm,marginparwidth=1.75cm]{geometry}

\usepackage{amsmath}
\usepackage{graphicx}
\usepackage{booktabs}

\usepackage[colorinlistoftodos]{todonotes}

\title{On the Control Effort in Output Feedback Sliding Mode Control of Sampled-Data Systems}
\author[1]{Thang Nguyen}
\author[2]{Christopher Edwards}
\author[3]{Vahid Azimi}
\author[4,*]{Wu-Chung Su}
\affil[1]{Department of Mechanical Engineering, Northern Arizona University, Flagstaff, Arizona, USA}
\affil[2]{College of Engineering, Mathematics and Physical Sciences, University of Exeter, Exeter, UK, EX4 4QF, UK}
\affil[3]{School of Electrical and Computer Engineering, Georgia Institute of Technology, Atlanta, GA, USA}
\affil[4]{Department of Electrical Engineering, National Chung-Hsing University, 250 Kuo-Kuang Road, Taichung, Taiwan, Republic of China}
\affil[*]{Corresponding author: email: wcsu@dragon.nchu.edu.tw}
\date{}

\begin{document}
\maketitle
\thispagestyle{fancy}

\begin{abstract}

In this paper, the problem of output feedback sliding mode control of linear sampled-data multi-input multi-output systems is considered.  Existing sliding mode control schemes can attenuate the influence of an external disturbance by driving system states onto a sliding surface. However, they can exhibit high gains during transients, which can be $O(1/T)$ where $T$ is the sampling time period. To address this problem, a new strategy, which employs disturbance approximation, is proposed so that the control effort will be $O(1)$. The new method avoids deadbeat phenomena and hence, it will be less sensitive to noise. Theoretical analysis is provided to show the convergence and robustness of the proposed method. Simulations were conducted to show the efficiency of the proposed approach.
\end{abstract}

\section{INTRODUCTION}

In practice, under the influence of disturbances and uncertainties, control performance of a dynamical system can deteriorate. To overcome this problem, sliding mode control has proved a powerful tool to reject disturbances and uncertainties using discontinuous control action with infinite switching frequency \cite{edwards1998}. This is applicable for continuous-time systems. Nowadays, the extensive use of digital devices in control systems necessiates the study of sample/hold effects when designing a control algorithm. Due to hardware limit, there is no control action with such infinite switching frequency as in continuous-time systems. Whilst, theoretically, sliding mode control has the ability to reject matched external disturbances or uncertainties for continuous-time systems \citep{DRAZENOVIC1969,edwards1998}, an ideal sliding mode cannot be obtained in \emph{sampled-data systems} due to the sampling/hold effect. In this situation, only ``quasi sliding modes'' are achieved, i.e. the system state is kept in a boundary layer around the sliding surface \citep{milo85}.

Numerous research works have been conducted addressing the problem of state feedback sliding mode control of sampled-data systems; see \citep{su,abxu07,Du2016,Niu_IETCTA_2010,Xu_IETCTA_2013,Behera_IETCTA_2015} and the references therein. The most common feature is that control laws are chattering-free and maintain an $O(T^2)$ quasi-sliding motion. In \citep{su}, a non-switching control method for a class of \emph{sampled-data systems} was exploited to avoid the chattering phenomena during the quasi sliding mode phase. In the state feedback sliding mode control problem, a dead-beat type control law  based on the one-step delayed disturbance approximation method results in a quasi sliding mode boundary layer of thickness $O(T^2)$, where $T$ is the sampling period \citep{su}. With this accuracy of quasi sliding mode, the state is kept in ultimate $O(T)$ bound \citep{abxu07}. An $O(T^2)$ quasi sliding mode can be obtained in \emph{sampled-data systems} in the context of state feedback \citep{abxu07}.

In this note, we aim to address the output feedback sliding mode control problem for linear sampled-data multi-input multi-output systems in the presence of external disturbances. Some papers in the literature proposed several output feedback sliding mode control methods for sampled-data systems in \cite{Lin2010,lai07,nguyen09b,Milo2013, Nguyen2016}. The methods in \cite{Lin2010,Milo2013} were only proposed for single input single output systems, which limit their applications. Similarly, a minimum variance control scheme in \cite{mitic2004} was presented where a quasi-sliding mode with $O(T^3)$ accuracy was achieved for single-input single-output systems. In \citep{nguyen09b,Nguyen10,Nguyen2016}, output feedback sliding mode control schemes were proposed for multi-input multi-output systems to achieve quasi-sliding motion with boundary layers of $O(T^2)$ and $O(T^3)$ respectively. However, the control signals in \citep{nguyen09b,Nguyen10,Nguyen2016} are of order $O(1/T)$, which can be detrimental to system hardware such as actuators during transients or in the presence of disturbances. Moreover, these (effectively) high gain controllers can be sensitive to measurement noise, which deteriorates the control performance. In this paper, improved versions of the control schemes in \citep{nguyen09b,Nguyen2016} are proposed to avoid possible high gain control efforts. Our paper exploits sampled-data predictors to estimate disturbances.

The contributions of the paper are:
\begin{itemize}
\item[i)] to provide a control technique to reduce the high-gain control effect during transient while maintaining a certain level of desired performance.
\item[ii)] to provide theoretical analysis of the proposed scheme. Note that the preliminary results in \cite{nguyen2017improvement} still omit complete theoretical analysis due to space reason.
\item[iii)] to evaluate the performance of the proposed scheme across different cases. Unlike the conference version in \cite{nguyen2017improvement}, in this paper, we consider the influence of noise on the new method.
\end{itemize}

Note that in the conference version \cite{nguyen2017improvement}, the disturbance and its first and second derivatives are required to be bounded. Meanwhile, in this paper, we consider a more general case in which only the disturbance and its first derivative are bounded.

In this paper, $\lambda\{A\}$ represents the spectrum of the matrix $A$, while $I_m$ is the identity matrix of order $m$. A vector function $f(t,s)\in R^n$ is said to be $O(s)$ over an interval $[t_1,t_2]$, if there exist positive constants $K$ and an $s^*$ such that $\|f(t,s)\|\leq Ks, \quad \forall s\in[0,s^*],\quad \forall t\in[t_1,t_2]$ \citep{kok86}. Throughout the paper, $f[k]$ stands for $f(kT)$, where $k=0, 1, 2, ...$ describes the index of the discrete-time sequence.

The paper is organized as follows. Section II presents the formulation of the problem. The main results are described in Section III. Simulation results are implemented to illustrate the efficacy of the proposed schemes in Section IV. The final section offers some conclusions. 

\section{PROBLEM FORMULATION}

Consider the following system
\begin{eqnarray}\label{a-0}
    \dot{x}(t)&=&Ax(t)+B(u(t)+f(t))\\
    \nonumber y(t)&=&Cx(t),
\end{eqnarray}
where $x(t)\in R^n$ is the system state, $u(t)\in R^m$ is the system control input, $y(t)\in R^p$ is the system output, $f(t)\in R^m$ is an unknown bounded external disturbance, with $m\leq p<n$. A switching function based on output information will be considered:
\be\label{a}
         s=Hy.
\ee
\newtheorem{assumption}{Assumption}

\begin{assumption}\label{as1}
The disturbance $f(t)$ and its first derivative are bounded.
\end{assumption}
\begin{assumption}\label{as2}
The disturbance $f(t)$ and its first and second derivatives are bounded.
\end{assumption}
\begin{assumption}\label{as3}
There exists a full rank $m \times p$ matrix $H$ such that the square matrix $HCB$ is invertible and the continuous-time sliding surface, $s(t)=0$, is a legitimate design in the sense that the reduced order motion is stable \citep{edspu95}.
\end{assumption}

\begin{remark}
According to \citep{edspu95}, if system (\ref{a-0}) has relative degree equal to one with stable invariant zeros and $B$ and $C$ have full rank \citep{edspu95}, then Assumption \ref{as3} is satisfied. A method to design matrix $H$ can be based on the framework in \citep{edspu95}.
\end{remark}

\begin{remark}
In this paper, we consider a more general class of disturbances, which only requires the boundedness of the disturbance and its first derivative. In Assumption \ref{as2}, which is used in \cite{nguyen09b,Nguyen2016,nguyen2017improvement}, the disturbance and its first and second derivatives are bounded. The disturbance considered in \cite{su,abxu07} is smooth.
\end{remark}

The sampled-data version of (\ref{a-0}) is
\begin{align}\label{b-3a}
    x[k+1]=&\Phi x[k]+\Gamma u[k]+d[k]\\
 \nonumber y[k]=&Cx[k],
\end{align}
where
\bea
 	\Phi&=&e^{AT}=\sum_{k=0}^\infty \frac{(TA)^k}{k!},\\
 	\Gamma&=&\int_0^T e^{A\tau} d\tau B=\int_0^T\sum_{k=0}^\infty \frac{(\tau A)^k}{k!}Bd\tau,
\eea
and in (\ref{b-3a}) the disturbance is
\be\label{dk}
	d[k]=\int_0^T e^{A\tau}Bf((k+1)T-\tau)d\tau.
\ee
Define
\bea
	\bar{A}&=&\frac{1}{T}(\Phi-I_n),\\
	\bar{\bar{A}}&=&\frac{1}{T^2}(\Phi-I_n-TA)=\sum_{k=2}^{\infty}T^{k-2}\frac{A^k}{k!}=O(1),\\
	\bar{B}&=&\frac{\Gamma}{T},\\
	\bar{\bar{B}}&=&\frac{1}{T^2}(\Gamma-TB).
\eea
With the above definitions, the system matrices of the discrete-time system (\ref{b-3a}) satisfy
\bea
    \Phi&=&I_n+T\bar{A}=I_n+T(A+T\bar{\bar{A}}) \label{phi}\\
    \Gamma &=&T\bar{B}=T(B+T\bar{\bar{B}}) \label{gammaB}.
\eea

Due to the sampling effect, the disturbance $d[k]$ in the sampled-data system contains unmatched components: details of its properties are described in the following lemmas.

\begin{lemma}\label{l1}
If Assumption \ref{as1} holds, then
\begin{subequations}\label{lem-1}
\begin{align}
     d[k]&=\Gamma f[k]+d^\prime[k],\label{dk_lem1}\\
     d[k]&-d[k-1]=O(T^2) \label{dk2_lem1}
\end{align}
\end{subequations}
where
\be\label{dpk_l1}
	d^\prime[k]=\int_0^T e^{A\tau}B\int_{kT}^{(k+1)T-\tau}v(\beta)d\beta d\tau=O(T^2),
\ee
\begin{equation}
    v(t)=df(t)/dt.
\end{equation}
\end{lemma}

\begin{proof}
Consider $0\leq\tau<T$ and express $f((k+1)T-\tau)$ as
\be\label{f_lem1}
	f((k+1)T-\tau)=f[k]+\int_{kT}^{(k+1)T-\tau}v(\beta)d\beta.
\ee
Substituting (\ref{f_lem1}) into (\ref{dk}), we obtain
\bea
\nn	d[k]&=&\int_0^T e^{A\tau}B(f[k]+\int_{kT}^{(k+1)T-\tau}v(\beta)d\beta)d\tau\\
\nn		&=&\int_0^T e^{A\tau}Bf[k]d\tau+d^\prime[k]\\
		&=&\Gamma f[k]+d^\prime[k].
\eea
Assume $v(t)$ is bounded by $V$, namely $v(t)\leq V$. We have
\bea
\nn	\|d^\prime[k]\|&\leq& \|\int_0^T e^{A\tau}B\int_{kT}^{(k+1)T-\tau}v(\beta)d\beta d\tau\| \\
\nn			&\leq& \|\int_0^T e^{A\tau}B(T-\tau) d\tau\|\\
\nn			&\leq& \|\int_0^T e^{A\tau}B T d\tau\|\\
			&=&T\Gamma=O(T^2).
\eea

We have
\bea
\nn	d[k]-d[k-1]&=&\Gamma (f[k]-f[k-1])+(d^\prime[k]-d^\prime [k-1]\\
			&=&\Gamma \int_{kT}^{(k+1)T}v(t) dt+(d^\prime[k]-d^\prime [k-1].
\eea
Since
\be
	\|\int_{kT}^{(k+1)T}v(t) dt\|\leq \int_{kT}^{(k+1)T}V dt=TV=O(T),
\ee
$\Gamma=O(T)$, and $d^\prime[k]=O(T^2)$,
\be
\nn	d[k]-d[k-1]=O(T^2).	
\ee
\end{proof}

The following lemma was employed in  \citep{abxu07,nguyen09b,Nguyen10,Nguyen2016,nguyen2017improvement}.
\begin{lemma}\label{l2}
If Assumption \ref{as2} holds, then
\begin{subequations}\label{lem-2}
\begin{align}
     d[k]&=\Gamma f[k]+d^\prime[k] \label{dk_lem2}\\
     d[k]&-d[k-1]=O(T^2), \label{dk2_lem2}\\
     d[k]&-2d[k-1]+d[k-2]=O(T^3), \label{dk3_lem2}
\end{align}
\end{subequations}
where
\bea \label{dpk_l2}
	d^\prime[k]&=&\frac{T}{2}\Gamma v[k]+T^3\Delta d[k]=O(T^2),\\
\label{delta_dk}
     \nn\Delta d[k]&=&\hat{M}v[k]+\frac{1}{T^3}\int_0^T e^{A\tau}B\int_{kT}^{(k+1)T-\tau}\int_{kT}^{\beta}\dot{v}(\sigma)d \sigma d\beta d\tau\\
    &=&O(1),
\eea
Note that in (\ref{dk})
\begin{equation}
    \hat{M}=(-\frac{1}{12}A-\frac{T}{12}\bar{\bar{A}})B=O(1).
\end{equation}

\end{lemma}
\begin{proof}
The proof is presented in \cite{nguyen09b}.
\end{proof}

\begin{remark}
A discrete-time model can be derived using the delta operator in \citep{Middleton_TAC_1986} from which a switching sliding mode control scheme is proposed to address state feedback control for a discrete-time system subject to an external disturbance \citep{Kumari_2016_IECON}. Meanwhile, our control law is non-switching in the context of output feedback. The problem in \citep{Kumari_2016_IECON} is suitable for fast sampling rates with more simple assumptions; i.e., there are no unmatched disturbance components in the discrete-time model. Furthermore, our problem is more general in the sense that it is not limited to fast sampling.
\end{remark}

Following the derivation in \citep{Nguyen2016}, we employ the following nonsingular transformation matrix
\begin{equation}\label{P-1}
    P_1=\left[\begin{matrix}M\\HC\end{matrix}\right],
\end{equation}
where $M\in R^{(n-m)\times n}$ and $MB=0$, which implies $\textrm{Range}(M^T)=\textrm{Null}(B^T)$. As demonstrated in \citep{Nguyen2016}, $P_1$ has full rank. Let the inverse of $P_1$ be partitioned as
\be\label{invP1}
    P_1^{-1}=\left[\begin{matrix}Q&R\end{matrix}\right]
\ee
where $Q$ has $n-m$ columns. Let $\left[\begin{matrix}\xi^T&s^T\end{matrix}\right]^T=P_1x$, then in the new coordinates
\begin{equation}\label{a-2}
    \left[\begin{matrix}\dot{\xi}\\\dot{s}\end{matrix}\right]=\left[\begin{matrix}MAQ&MAR\\HCAQ&HCAR\end{matrix}\right]\left[\begin{matrix}\xi\\s\end{matrix}\right]
        +\left[\begin{matrix}0\\HCB\end{matrix}\right](u+f).
\end{equation}

This is in ``normal form'', which implies that the sliding mode dynamics of system (\ref{a-2}) is
\begin{equation}\label{Ac}
    \dot{\xi}=MAQ\xi=A_c\xi
\end{equation}
where the eigenvalues of matrix $A_c$ contains any invariant zeros of (\ref{a-0}), \citep{edspu95}.

Now, consider the sampled-data version of the continuous-time system in (\ref{a-0}):
\begin{align}\label{b-3}
    \nonumber x[k+1]=&\Phi x[k]+\Gamma u[k]+d[k]\\
                    y[k]=&Cx[k]\\
          \nonumber s[k]=&Hy[k],
\end{align}
where the output feedback sliding vector is prescribed in (\ref{a}). The control methods in \citep{nguyen09b,Nguyen2016} designed for system (\ref{b-3}) can exhibit high gain transients of the order of $O(1/T)$ when the system state is far from the sliding surface. In this paper, our objective is to provide a solution to this high gain problem such that the control efforts are $O(1)$, but a certain level of accuracy of sliding mode is still guaranteed.

\section{MAIN RESULTS}

In this section, the improved versions of the schemes in \citep{nguyen09b,Nguyen2016} will be proposed. For convenience, we call the method in \citep{nguyen09b} Method~1 (M1), and the one in \citep{Nguyen2016} Method 2 (M2).

Using (\ref{phi}), (\ref{gammaB}), (\ref{dk_lem1}), and (\ref{P-1}), the system in (\ref{b-3}) is
\begin{align}\label{b-4}
    \nn\left[\begin{matrix}\xi[k+1]\\s[k+1]\end{matrix}\right]=&\left[\begin{matrix}I_{n-m}+TM\bar{A}Q&TM\bar{A}R\\THC\bar{A}Q&I_m+THC\bar{A}R\end{matrix}\right]\left[\begin{matrix}\xi[k]\\s[k]\end{matrix}\right]\\
        &+\left[\begin{matrix}TM\bar{B}\\THC\bar{B}\end{matrix}\right]u[k]+\left[\begin{matrix}d_{11}[k]\\d_{12}[k]\end{matrix}\right]
\end{align}
where
\bea
    d_{11}[k]&=&T^2M\bar{\bar{B}} f[k]+M d^\prime[k]=O(T^2), \label{d11}\\
    d_{12}[k]&=&THC\bar{B} f[k]+HCd^\prime[k]=O(T), \label{d12}
\eea
since $\bar{B}=B+T\bar{\bar{B}}$, $MB=0$, $M\bar{B}=O(T)$, and $d^\prime[k]=O(T^2)$ according to Lemma \ref{l1}.

The $s[k]$ dynamics in (\ref{b-4}) can be written as
\begin{equation}\label{ss}
    s[k+1]=(I_m+T\Omega_2)s[k]+T HC\bar{B}u[k]+g[k],
\end{equation}
where
\bea\label{gk}
    g[k]&=&T\Omega_1\xi[k]+d_{12}[k],\\
\label{om1}\Omega_1&=&HC\bar{A}Q,\\
\label{om2} \Omega_2&=&HC\bar{A}R.
\eea

As in \citep{ud89}, solving for $s[k+1]=0$ yields the discrete-time equivalent control law
\begin{equation}\label{ueq}
    u^{eq}[k]=-\frac{1}{T}(HC\bar{B})^{-1}((I_m+T\Omega_2)s[k]+g[k]),
\end{equation}
which is not physically implementable since it contains $g[k]$, which is unknown at time instant $k$ \cite{Nguyen2016}. The cause of the high gain phenomenon stems from the fact that the control laws are designed to force $s[k+1]=0$. To mitigate this problem, we design a control law such that
\be\label{sk}
	s[k+1]=\alpha s[k]
\ee
where $|\alpha|<1$ is a design parameter. Solving (\ref{sk}) for the equivalent control law,
\be\label{ueq2}
	 u[k]=-\frac{1}{T}(HC\bar{B})^{-1}(((1-\alpha)I_m+T\Omega_2)s[k]+g[k]).
\ee

\begin{remark}
According to Assumption \ref{as3}, $HCB$ is nonsingular. Since $\bar{B}=B+O(T)$ by construction, there is a small enough $T$ such as $HC\bar{B}$ is invertible. This was proved in \citep{Nguyen2016}.
\end{remark}

When the system state is not close to the origin such that $\xi[k]=O(1)$ and $s[k]=O(1)$, the expression in (\ref{gk}) implies that $g[k]=O(T)$. Choose $\alpha\in(0,1)$ such that
\be\label{beta}
	\beta\triangleq\frac{1-\alpha}{T}=O(1).
\ee

The expression in (\ref{gk}) shows that $g[k]=O(T)$. From (\ref{beta}), $1-\alpha=\beta T$ and thus,  the equivalent control $u[k]$ in (\ref{ueq2}) is $O(1)$. Note that $g[k]$ is unknown at time instant $[k]$. In the following, we will present methods to approximate $g[k]$ to obtain a physically realiseable control law.

\subsection{Development of a Modified Version of Method 1}
\label{c1}

In this subsection, we consider the case when Assumption \ref{as1} holds. From (\ref{d12}) and (\ref{gk}), $g[k]$ contains $f[k]$, which can be approximated by $f[k-1]$ due to the continuity and boundedness properties of $f(t)$ and its first derivative. Here, $\|f[k]-f[k-1]\|=\|\int_{(k-1)T}^{kT}v(\beta)d\beta\|\leq \|\int_{(k-1)T}^{kT}V d\beta\|=TV=O(T)$ where $\|v(t)\|\leq V$. As shown in \citep{nguyen09b}, $g[k]$ can be approximated by $g[k-1]$ which is computed from (\ref{ss}) as
\begin{equation}\label{d-1b}
    g[k-1]=s[k]-(I_m+T\Omega_2)s[k-1]-THC\bar{B}u[k-1].
\end{equation}

Hence, using $g[k-1]$ for $g[k]$ in (\ref{ueq2}) yields the expression
\be\label{mu1}
	 u[k]=-\frac{1}{T}(HC\bar{B})^{-1}(((1-\alpha)I_m+T\Omega_2)s[k]+g[k-1])
\ee
From (\ref{d-1b}) and (\ref{mu1}),
\bea\label{mu1b}
\nn u[k]&=&-\frac{1}{T}(HC\bar{B})^{-1}(((2-\alpha)I_m+T\Omega_2)s[k]\\
   	&&-(I_m+T\Omega_2)s[k-1])+u[k-1].
\eea
From (\ref{beta}), $1-\alpha=\beta T$ and (\ref{mu1})
\bea
\nn u[k]&=&-\frac{1}{T}(HC\bar{B})^{-1}((\beta T I_m+T\Omega_2)s[k]+g[k-1])\\
\nn &=&-(HC\bar{B})^{-1}((\beta I_m+\Omega_2)s[k]
		+\frac{g[k-1]}{T})\\
&=&O(1).\label{mu1c}
\eea

Next, we study the stability of the closed-loop system under the control law (\ref{mu1b}) in the absence of external disturbances. Let
\be
 \label{psi1}   \psi_1[k]=\left[\begin{matrix} \xi[k]\\ s[k]\\ \gamma[k]\end{matrix}\right],
\ee
where
\be
\label{u-1}
                \gamma[k]=T HC\bar{B}u[k].
\ee

From (\ref{ss}), (\ref{gk}), (\ref{mu1b}) and (\ref{u-1}),
\bea\label{gamma_dynamics}
\nn	\gamma[k+1]&=&-((2-\alpha)I_m+T\Omega_2)s[k+1]+(I_m+T\Omega_2)s[k])+\gamma[k]\\
\nn			&=&-((2-\alpha)I_m+T\Omega_2)((I_m+T\Omega_2)s[k]+\gamma[k]+g[k])\\
\nn			&&	+(I_m+T\Omega_2)s[k])+\gamma[k]\\
\nn			&=&-((1-\alpha)I_m+T\Omega_2)((I_m+T\Omega_2)s[k]\\
\nn			&&-((1-\alpha)I_m+T\Omega_2)\gamma[k]\\
			&&-((2-\alpha)I_m+T\Omega_2)(T\Omega_1\xi[k]+d_{12}[k]).
\eea

From (\ref{b-4}), (\ref{psi1}), and (\ref{gamma_dynamics}), the dynamics of the closed-loop system utilizing the control law (\ref{mu1b}) is described by the augmented system
\begin{equation}\label{aug}
   \psi_1[k+1]=A_{aug1}  \psi_1[k]+d_2[k],
\end{equation}
where the system matrix
\begin{equation}\label{A-aug1}
    A_{aug1}=\left[\begin{matrix}A_s&TN_1\\TN_2&A_{e1}\end{matrix}\right]
\end{equation}
and the sub-matrix
\be\label{As}
    A_s=I_{n-m}+TM\bar{A}Q=I_{n-m}+TA_c+T^2M\bar{\bar{A}}Q
\ee
with $A_c$ as given in (\ref{Ac}), and
\small
\begin{align}\label{Ae1}
A_{e1}=
\left[\bmx(I_m+T\Omega_2)&I_m\\
    -((1-\alpha)I_m+T\Omega_2)(I_m+T\Omega_2)&-((1-\alpha)I_m+T\Omega_2)\emx\right].
\end{align}
The (augmented) disturbance term in system (\ref{aug}) is
\begin{equation}\label{d-2}
    d_2[k]=\left[\begin{matrix}d_{11}[k]\\ d_{12}[k]\\-((2-\alpha)I_m+T\Omega_2)d_{12}[k]\end{matrix}\right],
\end{equation}
and the off-diagonal matrices in (\ref{A-aug1}) are
\bea
\nn N_1&=&[M\bar{A}R\quad M\bar{\bar{B}}(HC\bar{B})^{-1}],\\
\nn N_2&=&\left[\begin{matrix}\Omega_1\\-((2-\alpha)I_m+T\Omega_2)\Omega_1\end{matrix}\right].
\eea

Before demonstrating stability of the closed-loop system in the absence of disturbances, we need the following:
\begin{lemma}\label{eigAe1}
The eigenvalues of $A_{e1}$ are $\alpha$ and $0$.
\end{lemma}

\proof
Using column operations,
\bea
\nn\lambda \{A_{e1}\}&=&\det\Big[{\begin{smallmatrix} \lambda I_m-(I_m+T\Omega_2)&-I_m\\ ((1-\alpha)I_m+T\Omega_2)(I_m+T\Omega_2)&\lambda I_m+((1-\alpha)I_m+T\Omega_2)\end{smallmatrix} \Big]}\\
\nn&=&(-1)^m\det\Big[{\begin{smallmatrix} -I_m&\lambda I_m-(I_m+T\Omega_2)\\ \lambda I_m+((1-\alpha)I_m+T\Omega_2)&((1-\alpha)I_m+T\Omega_2)(I_m+T\Omega_2)\end{smallmatrix} \Big]}\\
\nn &=&(-1)^m\det [((1-\alpha)I_m+T\Omega_2)(I_m+T\Omega_2)+(\lambda I_m\\
\nn&& \quad+((1-\alpha)I_m+T\Omega_2))(\lambda I_m-(I_m+T\Omega_2))]\\
&=& (-1)^m\det[ \lambda (\lambda-\alpha)I_m].
\eea
This proves the lemma.
\endproof

Since $0<\alpha<1$, the eigenvalues of $A_{e1}$ lie in the unit circle and we have the following theorem.

\begin{theorem}\label{th1}
Suppose Assumption \ref{as1} holds. In the absence of disturbances, under the discrete-time output feedback control law (\ref{mu1b}), the sampled-data system (\ref{b-4}) is asymptotically stable if the sampling period $T$ is small enough.
\end{theorem}

\proof
Using similar arguments to those in \citep{kato1995}, the eigenvalues of $A_{aug1}$ are
\bea
	\lambda_1&=&\lambda \{A_s+O(T^2)\}\\
	\lambda_2&=&\lambda \{A_{e1}+O(T^2)\}.
\eea
Since $A_c$ contains the stable eigenvalues associated with the zero dynamics of the original continuous-time sliding motion in (\ref{Ac}),  $\lambda\{A_s\}$ lie in the unit circle. Hence, according to \citep{Nguyen2016}, the eigenvalues of $A_{aug1}$  lie in the unit circle for sufficiently small $T$, which implies the stability of the closed-loop system.
\endproof

Next, the accuracy of the quasi-sliding motion of the system under the proposed control law (\ref{mu1b}) in the presence of the external disturbance will be studied.
Under the control law (\ref{mu1}),
\bea\label{sk2}
\nn	s[k+1]&=&\alpha s[k]+g[k]-g[k-1]\\
\nn &=&\alpha s[k]+T\Omega_1 (\xi[k]-\xi[k-1])\\
&&+d_{12}[k]-d_{12}[k-1].
\eea
Since $d_{12}[k]-d_{12}[k-1]= O(T^2)$ \citep{Nguyen2016},
\be
	s[k+1]=\alpha s[k]+T^2 M\bar{A}R s[k-1] +O(T^2)
\ee
At steady state, $s[k+1]\approx s[k]$ and
\be
	(1-\alpha-T^2 M\bar{A}R ) s[k] = O(T^2)
\ee
or
\be
	s[k] =(\beta I_m-T M\bar{A}R )^{-1} O(T)=O(T).
\ee

Similarly, at steady state, $\xi[k+1]\approx \xi[k]$ and from (\ref{b-4}),
\be
	\xi[k]= (TM\bar{A}Q)^{-1} O(T^2)=O(T).
\ee
Therefore,
\be
	x[k]=P_1^{-1}\left[ \begin{matrix}\xi[k]\\ s[k]\end{matrix}\right]=O(T).
\ee
The above analysis is summarized in the following theorem.

\begin{theorem}\label{th2}
Under Assumptions \ref{as1} and \ref{as3}, the sampled-data output feedback control law (\ref{mu1b}) produces a quasi-sliding motion about the sliding surface $s(t)$ with an $O(T)$ boundary layer and an ultimate bound of $O(T)$ on the original state variables. Furthermore, the control input is guaranteed to be $O(1)$ when the initial state variables are such that $s[0]=O(1)$.
\end{theorem}

\begin{remark}
The proposed control contains no switching actions, thereby avoiding chattering phenomena. On the other hand, it is observed that control law (\ref{mu1b}) is not able to completely compensate disturbance $g(k)$. However, by taking into account the past information, control law (\ref{mu1b}) still provides the closed-loop system with certain characteristics to reduce the influence of external disturbances.
\end{remark}

\subsection{Development of a Modified Version of Method 2}
\label{c2}

In this subsection, we consider the case when Assumption \ref{as2} holds. We have
\be
	f[k+1]=f[k]+v[k]T+f^\prime[k],
\ee
where
\be
	f^\prime[k]=\int_{kT}^{(k+1)T}\int_{kT}^\beta \dot{v}(\sigma)d\sigma d\beta.
\ee
Since the second derivative of $f(t)$ is bounded, assume that $\|\dot{v}\|\leq W=O(1)$. Hence,
\be
	\|\int_{kT}^{(k+1)T}\int_{kT}^\beta \dot{v}(\sigma)d\sigma d\beta\|\leq T(\beta-kT)W=O(T^2)
\ee
for $kT\leq \beta\leq (k+1)T$. Thus, $f^\prime[k]=O(T^2)$. We have
\bea
\nn	&&\|f[k]-2f[k-1]+f[k-2]\|\\
\nn &=&\|T(v[k-1]-v[k-2])+f^\prime[k-1]-f^\prime[k-2]\|\\
\nn				&=&\|T\int_{(k-2)T}^{(k-1)T} \dot{v}(\sigma)d\sigma\|\\
				&\leq&T^2W=O(T^2).
\eea
Therefore, $f[k]$ can be approximated by $2f[k-1]-f[k-2]$. Due to the expressions in (\ref{d12}) and (\ref{gk}), $g[k]$ can be approximated by $2g[k-1]-g[k-2]$. This approximation was also employed in the control law presented in \citep{Nguyen2016}. The expression in (\ref{d12}) also implies that
\be\label{d12_TO2}
	d_{12}[k]-2_{12}[k-1]+_{12}[k-2]=O(T^3).	
\ee

In the equivalent control law (\ref{ueq2}), replacing $g[k]$ by $2g[k-1]-g[k-2]$ yields
\bea\label{mu2}
\nn u[k]&=&-\frac{1}{T}(HC\bar{B})^{-1}(((1-\alpha)I_m+T\Omega_2)s[k]\\
		&&+2g[k-1]-g[k-2]).
\eea
Using (\ref{d-1b}), the control law in (\ref{mu2}) is
\bea\label{mu2b}
    \nn u[k]&=&-\frac{1}{T}(HC\bar{B})^{-1}(((3-\alpha)I_m+T\Omega_2)s[k]\\
\nn&&   -(3I_m+2T\Omega_2)s[k-1]+(I_m+T\Omega_2)s[k-2])\\
&&+2u[k-1]-u[k-2].
\eea
Using the same argument as in Subsection \ref{c1}, from (\ref{mu2}), we obtain
\bea\label{mu2c}
\nn u[k]&=&-(HC\bar{B})^{-1}((\beta I_m+\Omega_2)s[k]+\frac{2g[k-1]-g[k-2]}{T})\\
		&=&O(1).
\eea
As in Subsection \ref{c1}, we employ the following
\bea\label{ss-3}
    s_1[k]&=&s[k-1],\\
\label{u-1b}
    \gamma[k]&=&T HC\bar{B}u[k],\\
\label{u-2}
    \gamma_1[k]&=&T HC\bar{B}u[k-1].
\eea
Let
\be
    \psi_2=\left[\begin{matrix}\xi[k]\\s[k]\\s_1[k]\\\gamma[k]\\\gamma_1[k]\end{matrix}\right]
\ee
then the dynamics of the extended system is
\begin{equation}\label{aug2}
   \psi_2[k+1]=A_{aug2}   \psi_2[k]+d_3[k],
\end{equation}
where
\begin{equation}\label{A-aug}
    A_{aug2}=\left[\begin{matrix}A_s&TN_3\\TN_4&A_{e2}\end{matrix}\right],
\end{equation}
the sub-matrix $A_s$ is given in (\ref{As}), and
\be\label{Ae2}
A_{e2}=
\left[\begin{smallmatrix}(I_m+T\Omega_2)&0&I_m&0\\I_m&0&0&0\\
    -(-\alpha I_m+(2-\alpha)T\Omega_2+T^2\Omega_2^2)&-(I_m+T\Omega_2)&-((1-\alpha)I_m+T\Omega_2)&-I_m\\0&0&I_m&0\end{smallmatrix}\right].
\ee
The (augmented) disturbance term in system (\ref{aug2}) is
\be\label{d-3}
 \nn   d_3[k]=\left[\begin{matrix}d^T_{11}[k]& d^T_{12}[k]&0&-d^T_{12}[k]((3-\alpha)I_m+T\Omega_2)^T&0\end{matrix}\right]^T,
\ee
and the off-diagonal matrices in (\ref{A-aug}) are
\bea
\nn    N_3&=&[M\bar{A}R\quad 0_{(n-m)\times m}\quad M\bar{\bar{B}}(HC\bar{B})^{-1}\quad 0_{(n-m)\times m}],\\
\nn    N_4&=&\left[\begin{matrix}\Omega_1\\0\\-((3-\alpha)I_m+T\Omega_2)\Omega_1\\0\end{matrix}\right].
\eea

As argued in Subsection \ref{c1}, the following results are obtained.

\begin{lemma}\label{eigAe2}
The eigenvalues of $A_{e2}$ are $\alpha$ and $0$.
\end{lemma}
\proof
Let
\be
	X=\left[\begin{smallmatrix}(\lambda-1)I_m-T\Omega_2&0&-I_m&0\\-I_m&\lambda I_m&0&0\\
    (-\alpha I_m+(2-\alpha)T\Omega_2+T^2\Omega_2^2)&(I_m+T\Omega_2)&(\lambda+1-\alpha)I_m+T\Omega_2&I_m\\0&0&-I_m&\lambda I_m\end{smallmatrix}\right]
\ee
Then
\be
	\lambda \{A_{e2}\}=\det X.
\ee
Let
\bea
\nn X_{11}&=&(\lambda-1)I_m-T\Omega_2,\\
\nn X_{31}&=&-\alpha I_m+(2-\alpha)T\Omega_2+T^2\Omega_2^2,\\
\nn X_{32}&=&I_m+T\Omega_2,\\
\nn X_{33}&=&(\lambda+1-\alpha)I_m+T\Omega_2.
\eea
Furthermore, let
\be
	J_1= \left[\begin{matrix}I_m&0&0&0\\I_m&X_{11}&0&0\\
    0&0&I_m&0\\0&0&0&I_m\end{matrix}\right],
\ee
\be
	J_2= \left[\begin{matrix}I_m&0&0&0\\0&I_m&0&0\\
    -X_{31}&0&X_{11}&0\\0&0&0&I_m\end{matrix}\right],
\ee
\be
	J_3= \left[\begin{matrix}I_m&0&0&0\\0&I_m&0&0\\
    0&-X_{11}X_{32}&\lambda X_{11}&0\\0&0&0&I_m\end{matrix}\right],
\ee
\be
	J_4= \left[\begin{matrix}I_m&0&0&0\\0&I_m&0&0\\
    0&0&\lambda X_{11}&0\\0&0&I_m&Y_{33}\end{matrix}\right],
\ee
where
\be
	Y_{33}=X_{11}X_{32}+\lambda X_{11}(X_{11} X_{33}+X_{31}).
\ee
Then we have
\be
	J_4J_3J_2J_1X=\left[\begin{smallmatrix}X_{11}&0&-I_m&0\\0&\lambda X_{11}&-I_m&0\\
   0&0&Y_{33}&\lambda X_{11}^2\\0&0&0&\lambda Y_{33}+\lambda X_{11}^2\end{smallmatrix}\right].
\ee
Thus,
\be
	\det(J_4J_3J_2J_1X)=\det(X_{11})\det(\lambda X_{11})\det(Y_{33})\det(\lambda Y_{33}+\lambda X_{11}^2).
\ee
Since
\be
	\det(\lambda Y_{33}+\lambda X_{11}^2)=\det(\lambda^3(\lambda-\alpha)X_{11}),
\ee
\be
	\det(X)=\lambda^{3m}(\lambda-\alpha)^m.
\ee
This proves the lemma.
\endproof

\begin{theorem}\label{th1b}
Suppose Assumption \ref{as2} holds. In the absence of disturbances, under the discrete-time output feedback control law (\ref{mu2b}), the sampled-data system (\ref{b-4}) is asymptotically stable if the sampling period $T$ is small enough.
\end{theorem}
\proof
Using the results in perturbation theory for linear operators in \citep{kato1995}, the eigenvalues of $A_{aug2}$ are
\bea
	\lambda_1&=&\lambda \{A_s+O(T^2)\}\\
	\lambda_2&=&\lambda \{A_{e2}+O(T^2)\}.
\eea
According to Lemma \ref{eigAe2}, the eigenvalues of $A_{e2}$ are $\alpha$ and 0, which lie in the unit circle. In addition, the eigenvalues of $A_s$ also are in the unit circle. Hence, similar to the proof of Theorem \ref{th1}, for sufficiently small $T$, the eigenvalues of $A_{aug2}$  lie in the unit circle, which implies the stability of the closed-loop system.
\endproof

\begin{theorem}\label{th2b}
Suppose Assumptions \ref{as2} and \ref{as3} hold. In the presence of the external disturbance, the sampled-data output feedback control law (\ref{mu2b}) produces a quasi-sliding motion on the sliding surface $s(t)$ with an $O(T^2)$ boundary layer and an ultimate bound of $O(T)$ on the original state variables. Furthermore, the control input is guaranteed to be $O(1)$ if the initial state variables are such that $s[0]=O(1)$.
\end{theorem}
\proof
Under the control law (\ref{mu2}),
\bea\label{sk2b}
\nn	s[k+1]&=&\alpha s[k]+g[k]-2g[k-1]+g[k-2]\\
\nn &=&\alpha s[k]+T\Omega_1 (\xi[k]-2\xi[k-1]+\xi[k-2])\\
&&+d_{12}[k]-2d_{12}[k-1]+d_{12}[k-2].
\eea
Due to (\ref{d12_TO2}) and at steady state, $\xi[k+1]\approx \xi[k]$ , $s[k+1]\approx s[k]$, we obtain
\be
	s[k+1]=\alpha s[k]+O(T^3).
\ee
This implies
\be
	s[k] = \frac{1}{1-\alpha}O(T^3)=O(T^2)
\ee
since $1-\alpha=O(T)$.
Similarly, at steady state, $\xi[k+1]\approx \xi[k]$ and from (\ref{b-4}),
\be
	\xi[k]= (TM\bar{A}Q)^{-1} O(T^2)=O(T).
\ee
Therefore,
\be
	x[k]=P_1^{-1}\left[ \begin{matrix}\xi[k]\\ s[k]\end{matrix}\right]=O(T).
\ee
\endproof

\begin{remark}
The quasi-sliding motions under the control laws (\ref{mu1}), (\ref{mu2b}) possess less accuracy than the ones in \citep{nguyen09b,Nguyen2016}, which are $O(T^2)$ and $O(T^3)$ respectively. (However, the bounds on the state variables are similar). The advantage of the proposed schemes is that the resulting control efforts are able to operate at a less demanding mode than their original counterparts in \citep{nguyen09b,Nguyen2016}.
\end{remark}

\begin{remark}
When $s[k]=O(1)$, the control laws in (\ref{mu1b}) and (\ref{mu2b}) are $O(1)$ as explained in (\ref{mu1c}) and (\ref{mu2c}) respectively. In contrast, the methods in \citep{nguyen09b,Nguyen2016} provide $O(1/T)$ control effort, which results in undesired high gain control. This shows the advantage of the proposed method in this paper.
\end{remark}

\section{SIMULATION}
In this section, for comparison we employ the same system as in \citep{Nguyen2016}, which is the lateral dynamics of an aircraft \cite{Srina78}. Its system matrices are given as
\bea
\nn A&=&\left[\bmx-3.79& 0.04 &-52 &0\\-0.14& -0.36 &4.24 &0\\0.06& -1 &-0.27 &0.05\\1 &0.06 &0 &0 \emx\right], \\
\nn B&=&\left[\bmx25& 9.83\\1.42& -4.2\\0.01 &0.05\\0& 0\emx\right],\\
 \nn C&=&\left[\bmx1&0&0&0\\0&1&0&0\\0&0&0&1\emx\right].
\eea
The invariant zero of the system is -0.1796. The matrix $H$ is
\be
\nn H=\left[\begin{matrix}0.035306& 0.082634 &0.076550\\0.011937& -0.210157 &  0.008324\end{matrix}\right],
\ee
which is constructed such that the assignable eigenvalue for the sliding mode is $-2$, \citep{Nguyen2016}. The initial condition is given by $x(0)=[-1,2,1,-2]^T$. The sampling period is $T=0.01s$. The disturbance vector is defined as
\be
 f(t)=\begin{cases}
\left[\begin{matrix}0\\0\end{matrix}\right], & \text{for } t\geq 0\\
\left[\begin{matrix}2\\-0.5\end{matrix}\right], & \text{for } 10\leq t<5\pi\\
\left[\begin{matrix}1+\sin(0.5t)\\0.5\cos(t)\end{matrix}\right],& \text{for } t\geq5\pi.
\end{cases}
\ee
This shows that the disturbance vector affects the system dynamics from $t=10s$ onwards. At $t=5\pi$, the second derivative of the disturbance does not exist. The parameters of the controllers are $\beta=3$ and $\alpha=0.97$.

Since the control methods use information from previous time instants, the control signals from M1 and modified M1 (MM1) are activated only from time step 2 onwards while those of M2 and modified M2 (MM2) start working from time step 3. We conducted two experiments: a noise free one and a noisy case. (In the conference version of the paper \cite{nguyen2017improvement}, only the noise free case was considered).

Figs. \ref{figu1} and \ref{figu2} reveal that the modified versions produce control signals of less magnitude than the ones using the original control methods. Specifically, the largest magnitudes during the transient of the control law using M1 and M2 are about 23 and 25 respectively; meanwhile, the control effort using method MM1 and MM2 generates signals whose values are about 2. This suggests that the proposed schemes are able to maintain the control signals at low gain levels. At $t=10s$ and $t=5\pi s$, the control signal using method MM1 exhibits less fluctuation than that using method MM2. This suggests that method MM1 is less sensitive to distubances than method MM2. In Figs. \ref{figx1} and \ref{figx2}, the evolution of the state variables using M1 and M2 is slightly better than that using MM1 and MM2. The sliding functions are presented in Figs. \ref{figs1} and \ref{figs2} showing that M1 and M2 perform better than their counterparts. These numerical results illustrate our theoretical analysis.

A noise profile is added to the outputs of the system in the form of a uniformly distributed random signal, whose range lies in the interval $[-0.005,0.005]$. It is shown in Figs. \ref{figu1n}, \ref{figu2n} that MM1 and MM2 generate less control effort than M1 and M2. The largest magnitudes of the control signals using M1 and M2 are again about 24 and 23 respectively. In contrast, the magnitudes of the control efforts using MM1 and MM2 are about 3 and 4 respectively. Figs. \ref{figx1n}, \ref{figx2n}, \ref{figs1n}, \ref{figs2n} show that the evolution of the state variables and sliding functions using MM1 and MM2 is much better than their counterparts. MM1 performs best in this scenario as its control signals are less sensitive to noise than the others. It should also be observed that the performance of MM2 and M1 is comparable.

In both cases, the numerical simulations reveal that the proposed schemes are effective in avoiding high gain control efforts. In the absence of noise, M1 and M2 perform better than MM1 and MM2, but in contrast, in the presence of noise, MM1 and MM2 outperform their counterparts.

\begin{figure}[!htb]
      \centering
      \includegraphics[height=2.5in]{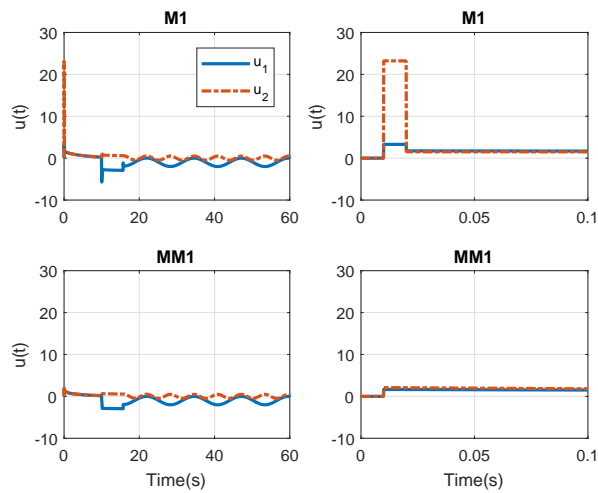}
       \caption{The evolution of the control signals using M1 and MM1 in the noise-free case}
      \label{figu1}
\end{figure}

\begin{figure}[!htb]
      \centering
      \includegraphics[height=2.5in]{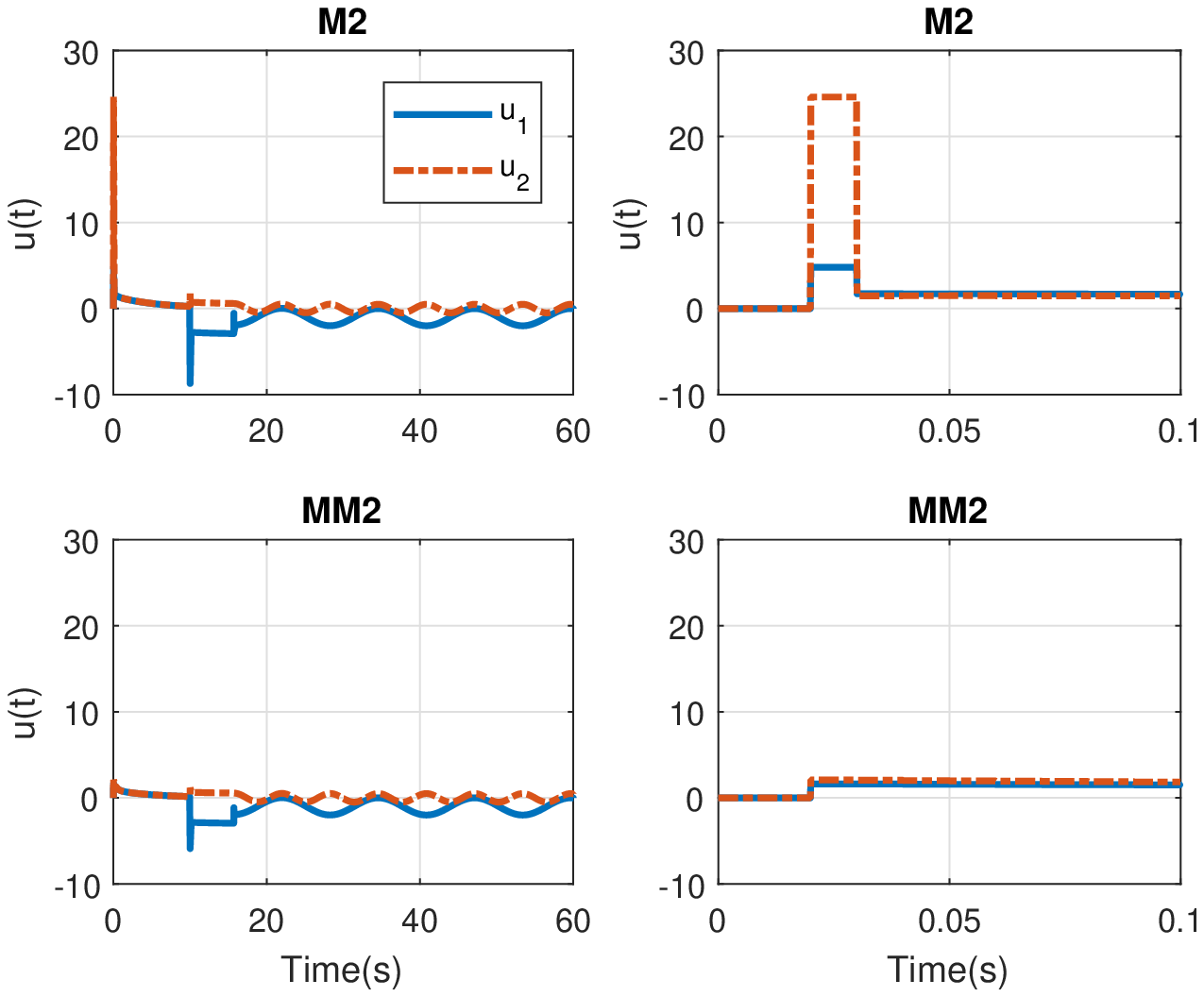}
      \caption{The evolution of the control signals using M2 and MM2 in the noise-free case}
      \label{figu2}
\end{figure}

\begin{figure}[!htb]
      \centering
      \includegraphics[height=2.5in]{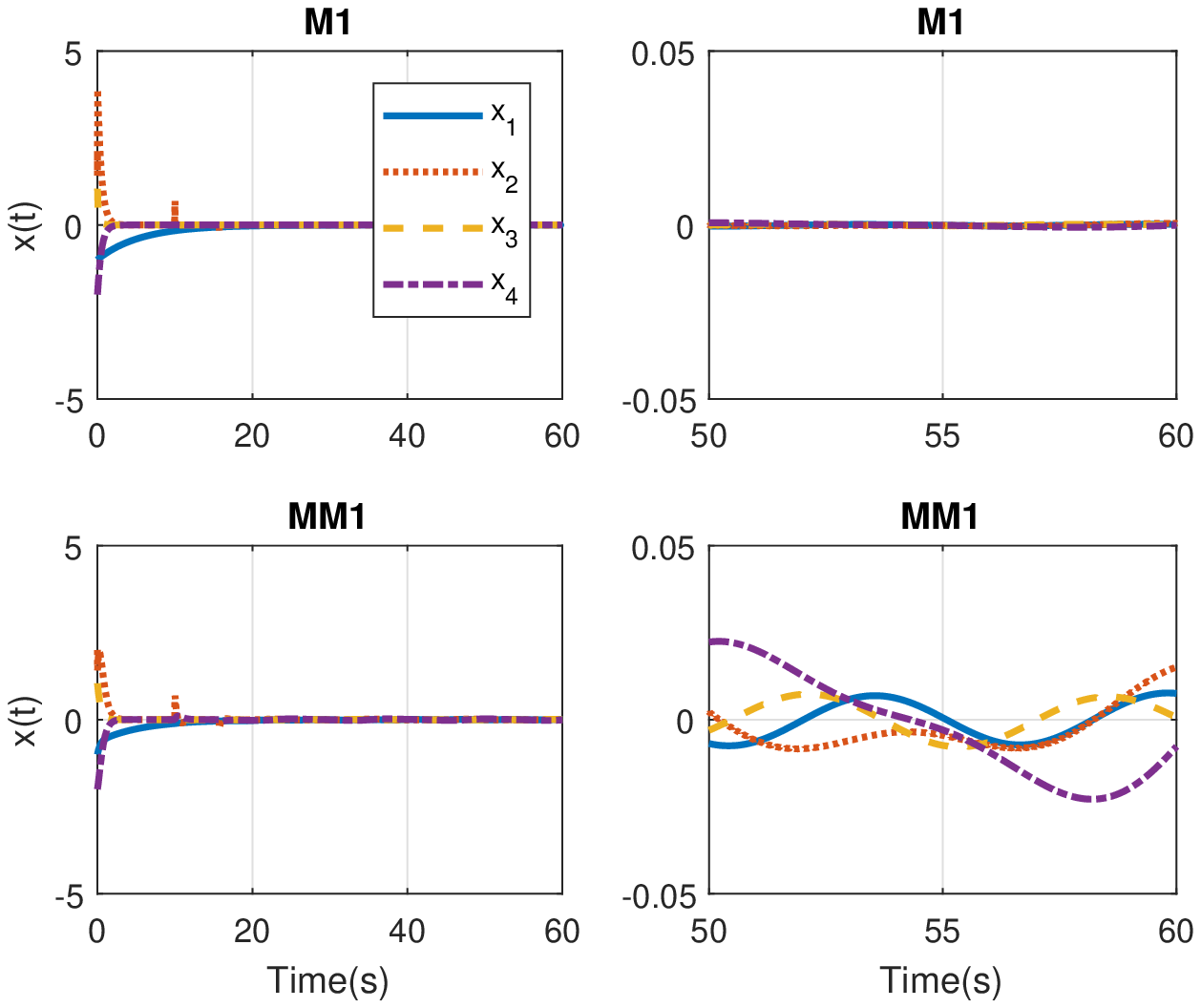}
      \caption{The evolution of the state variables using M1 and MM1 for the noise-free case}
      \label{figx1}
\end{figure}

\begin{figure}[!htb]
      \centering
      \includegraphics[height=2.5in]{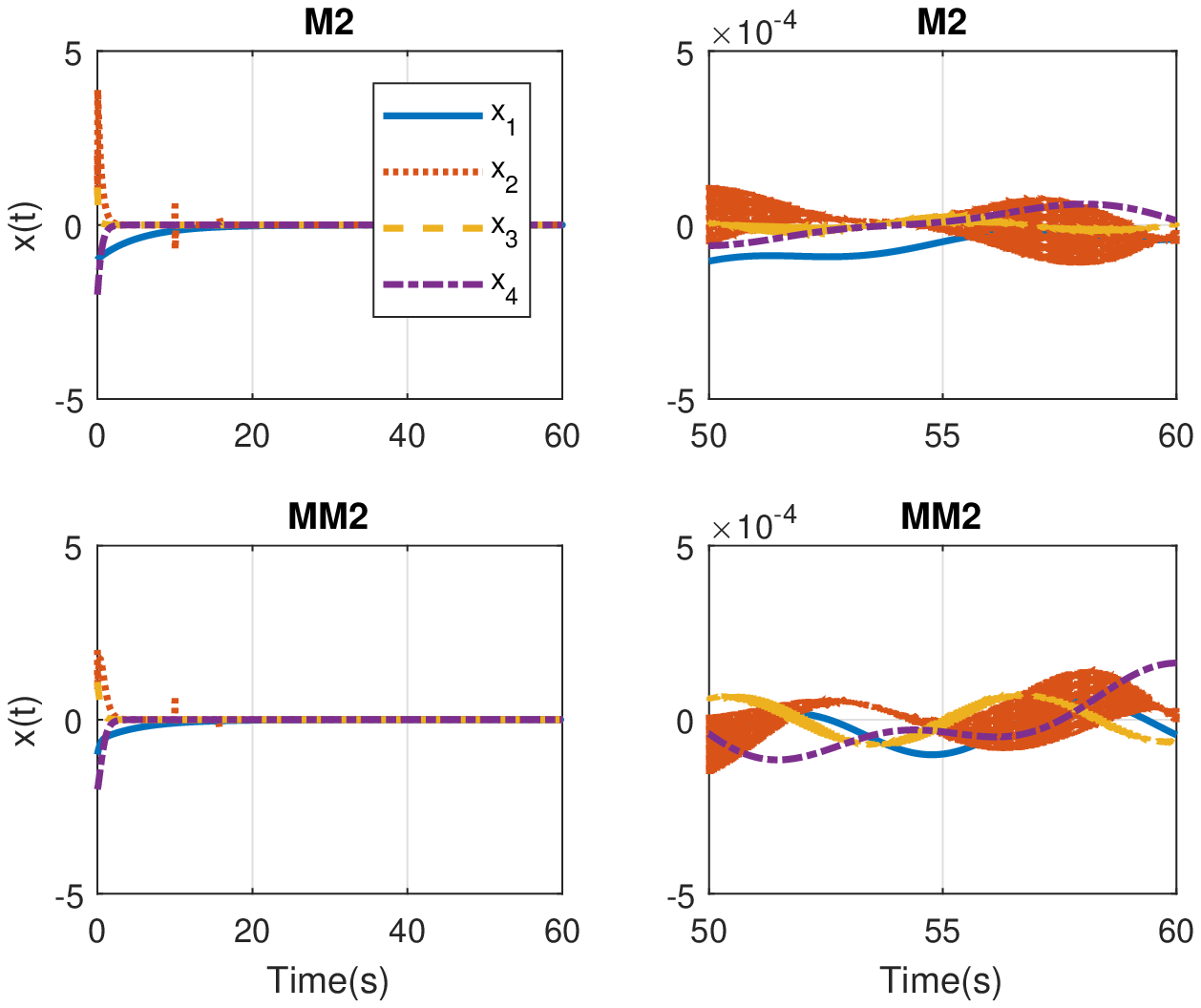}
       \caption{The evolution of the state variables using M2 and MM2 for the noise-free case}
      \label{figx2}
\end{figure}

\begin{figure}[!htb]
      \centering
      \includegraphics[height=2.5in]{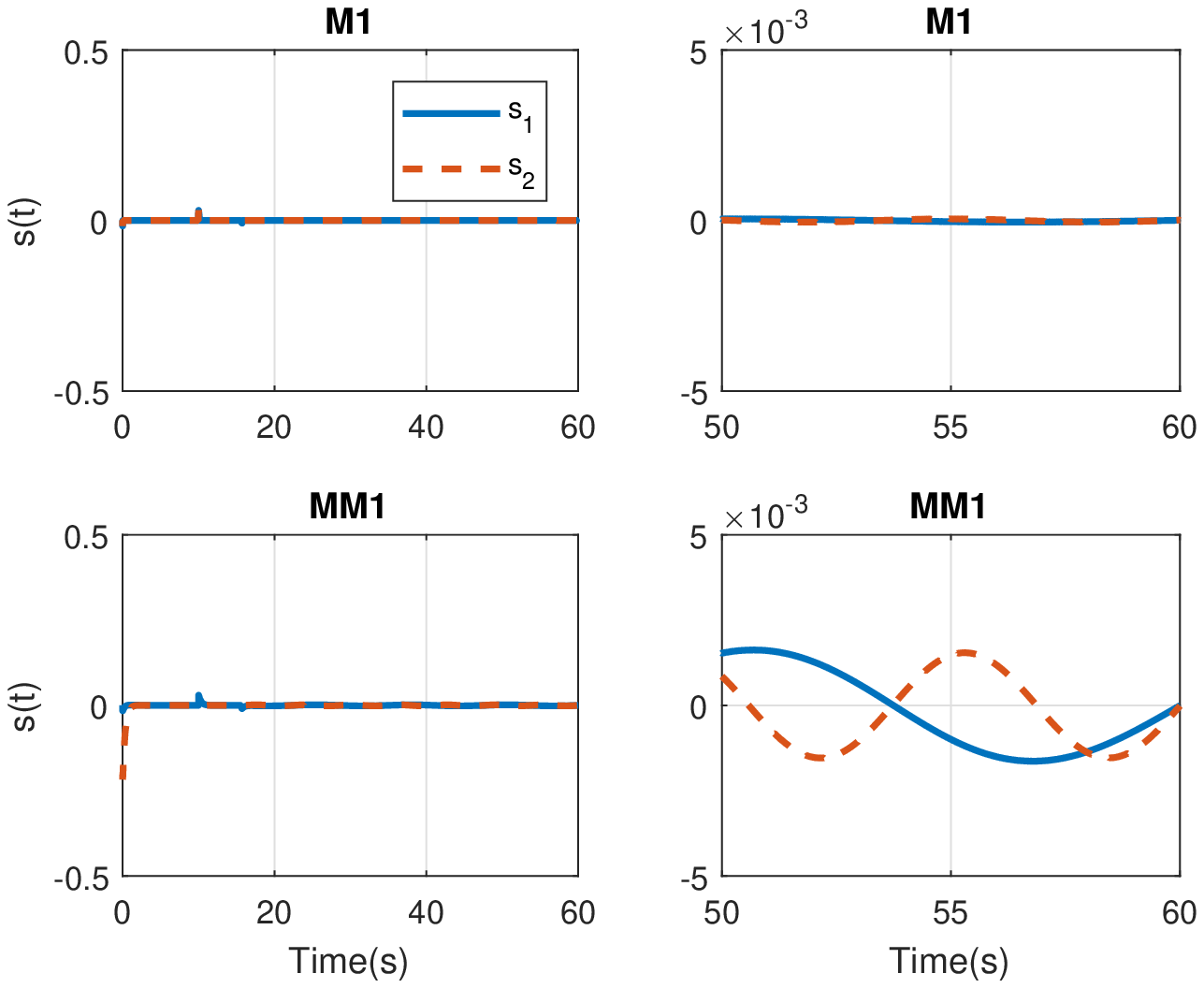}
      \caption{The evolution of the sliding functions using M1 and MM1 for the noise-free case}
      \label{figs1}
\end{figure}

\begin{figure}[!htb]
      \centering
      \includegraphics[height=2.5in]{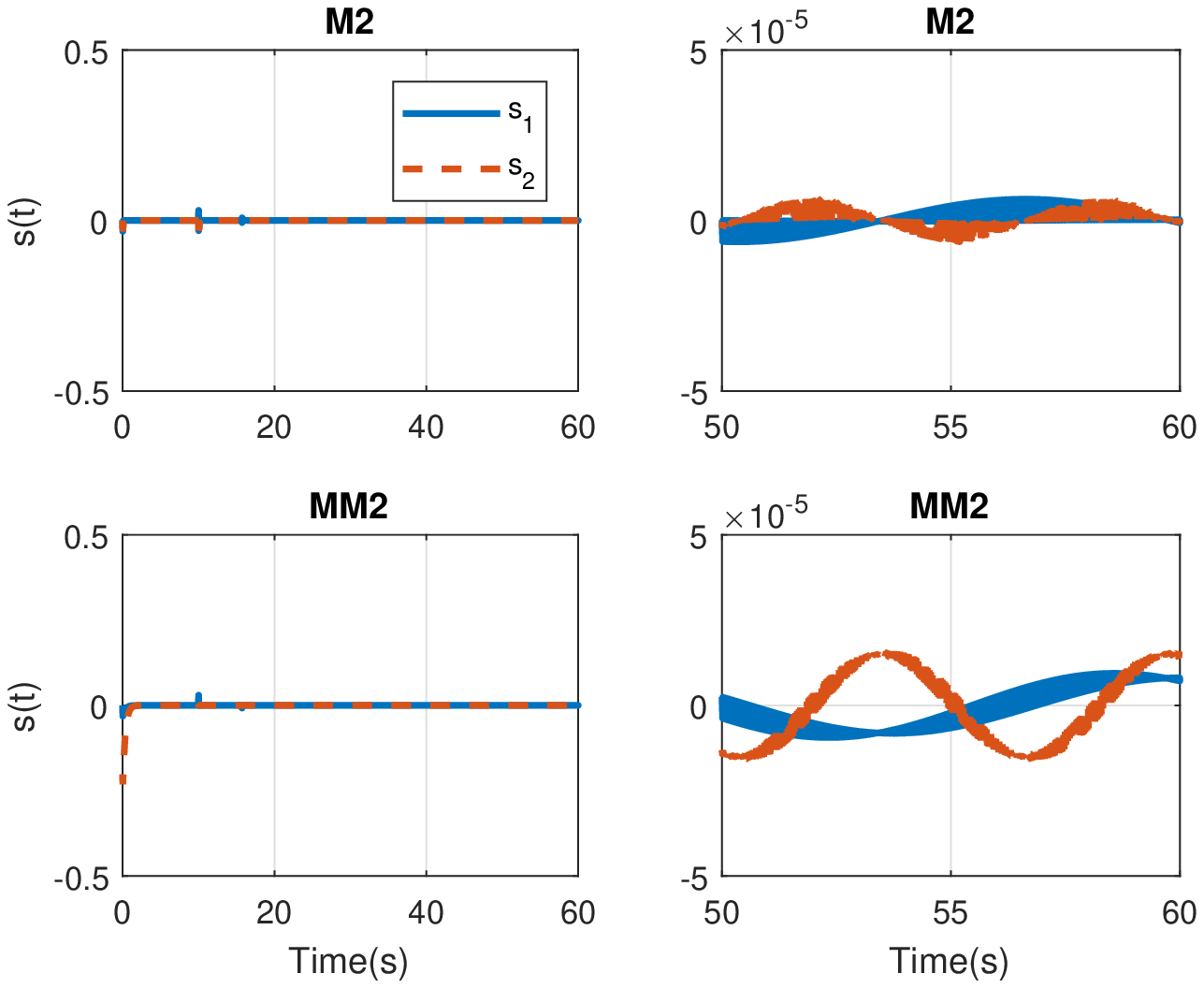}
       \caption{The evolution of the sliding functions using M2 and MM2 for the noise-free case}
      \label{figs2}
\end{figure}

\begin{figure}[!htb]
      \centering
      \includegraphics[height=2.5in]{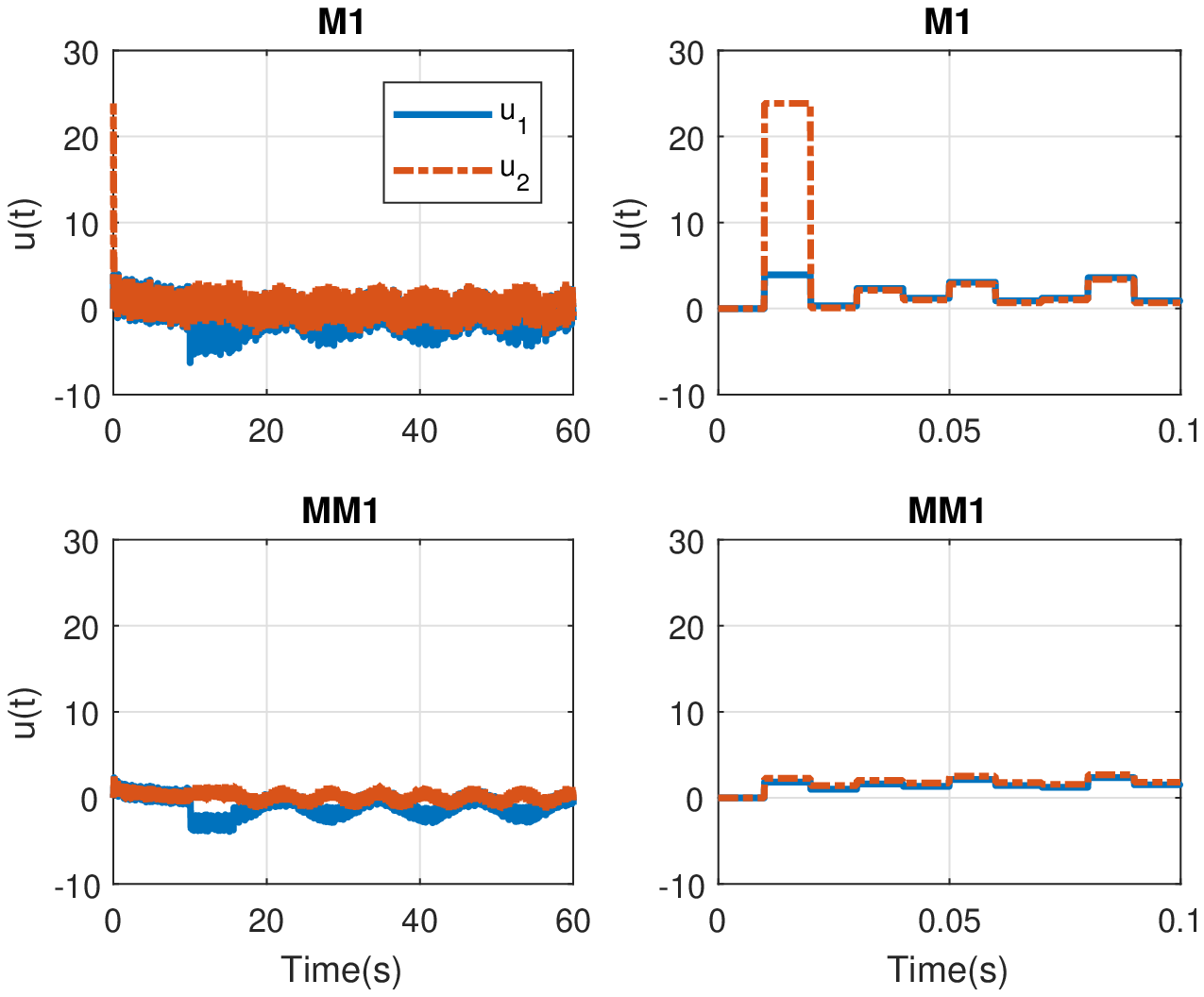}
      \caption{The evolution of the control signals using M1 and MM1 for the noisy case}
      \label{figu1n}
\end{figure}

\begin{figure}[!htb]
      \centering
      \includegraphics[height=2.5in]{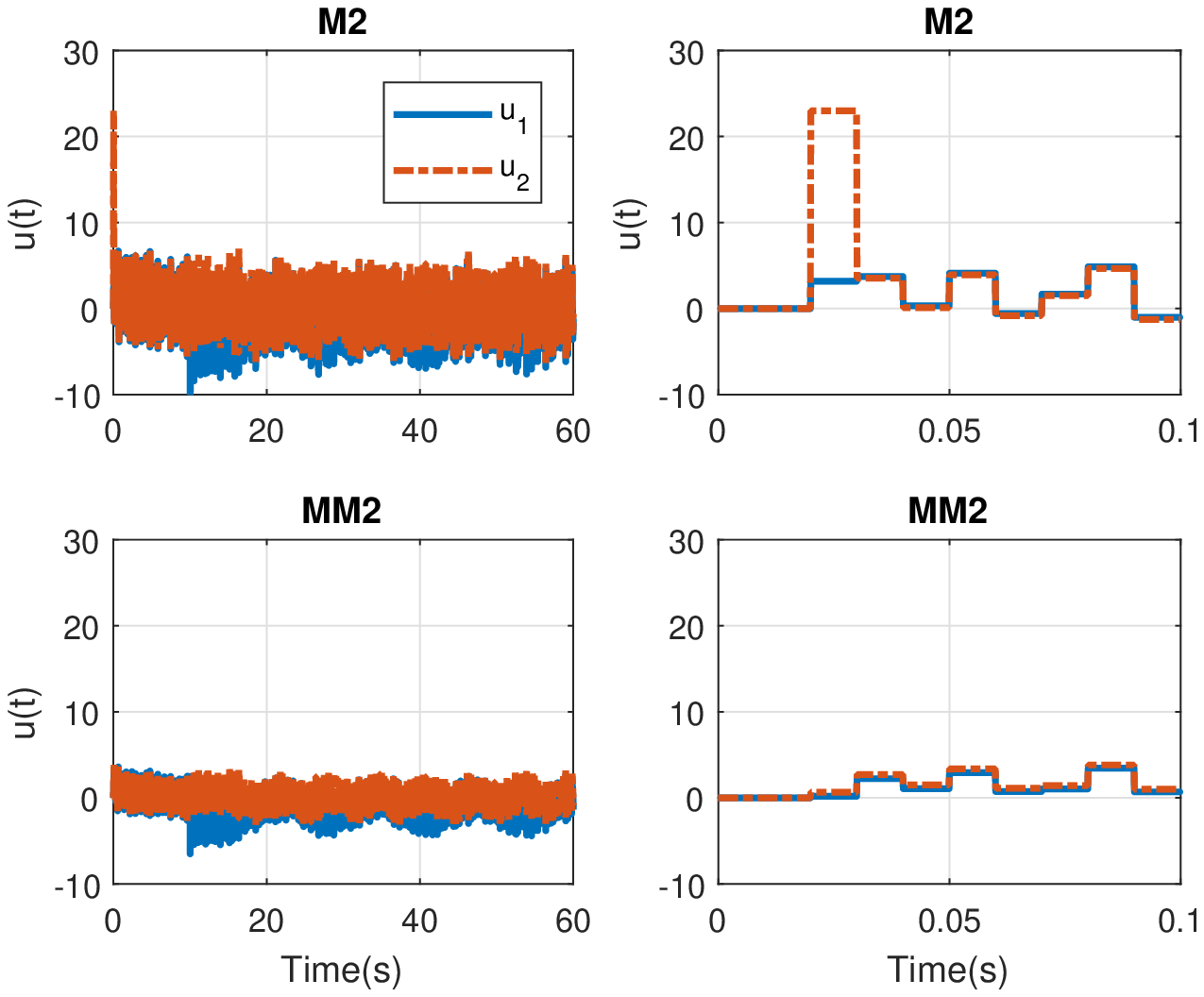}
      \caption{The evolution of the control signals using M2 and MM2 for the noisy case}
      \label{figu2n}
\end{figure}

\begin{figure}[!htb]
      \centering
      \includegraphics[height=2.5in]{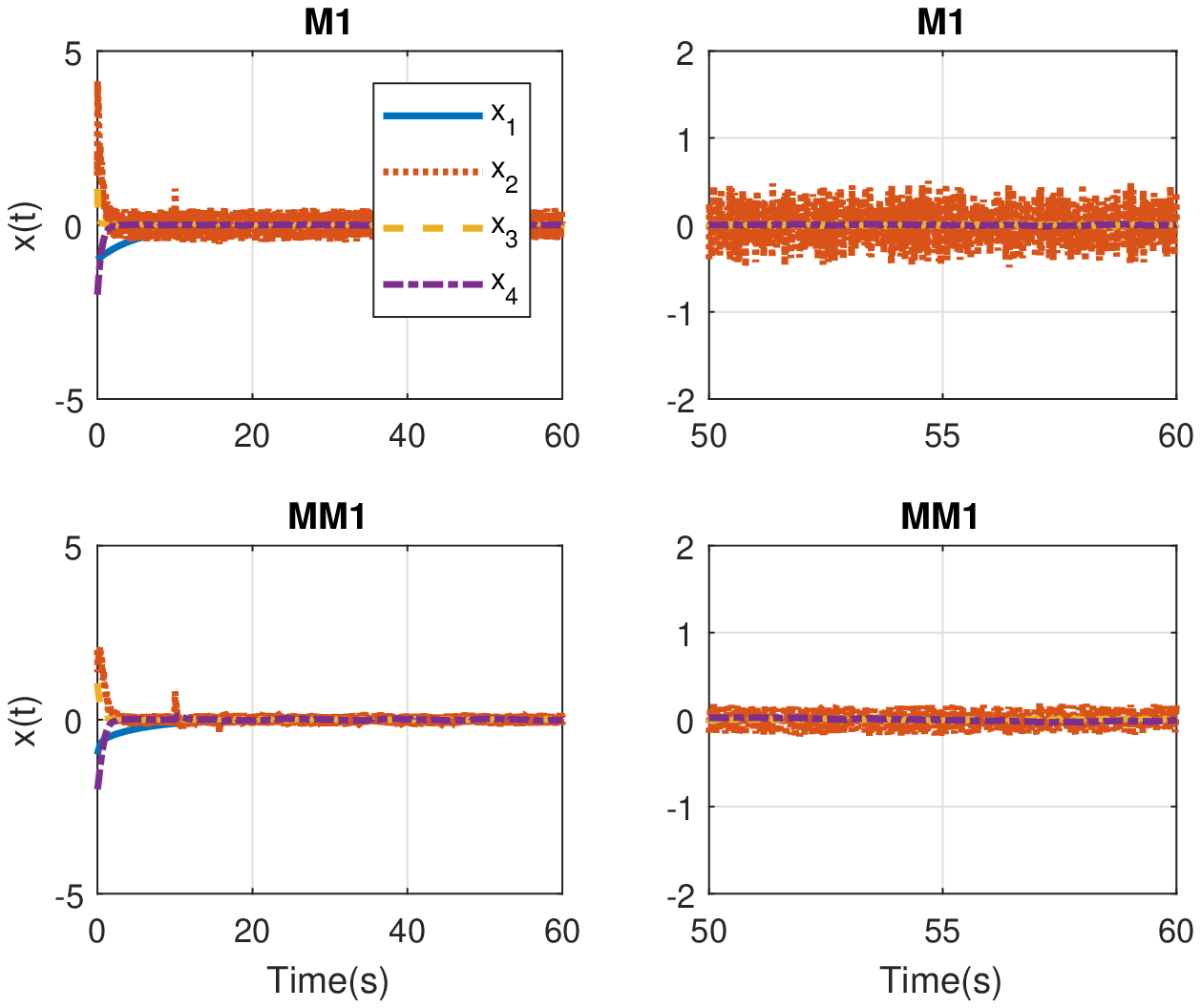}
      \caption{The evolution of the state variables using M1 and MM1 for the noisy case}
      \label{figx1n}
\end{figure}

\begin{figure}[!htb]
      \centering
      \includegraphics[height=2.5in]{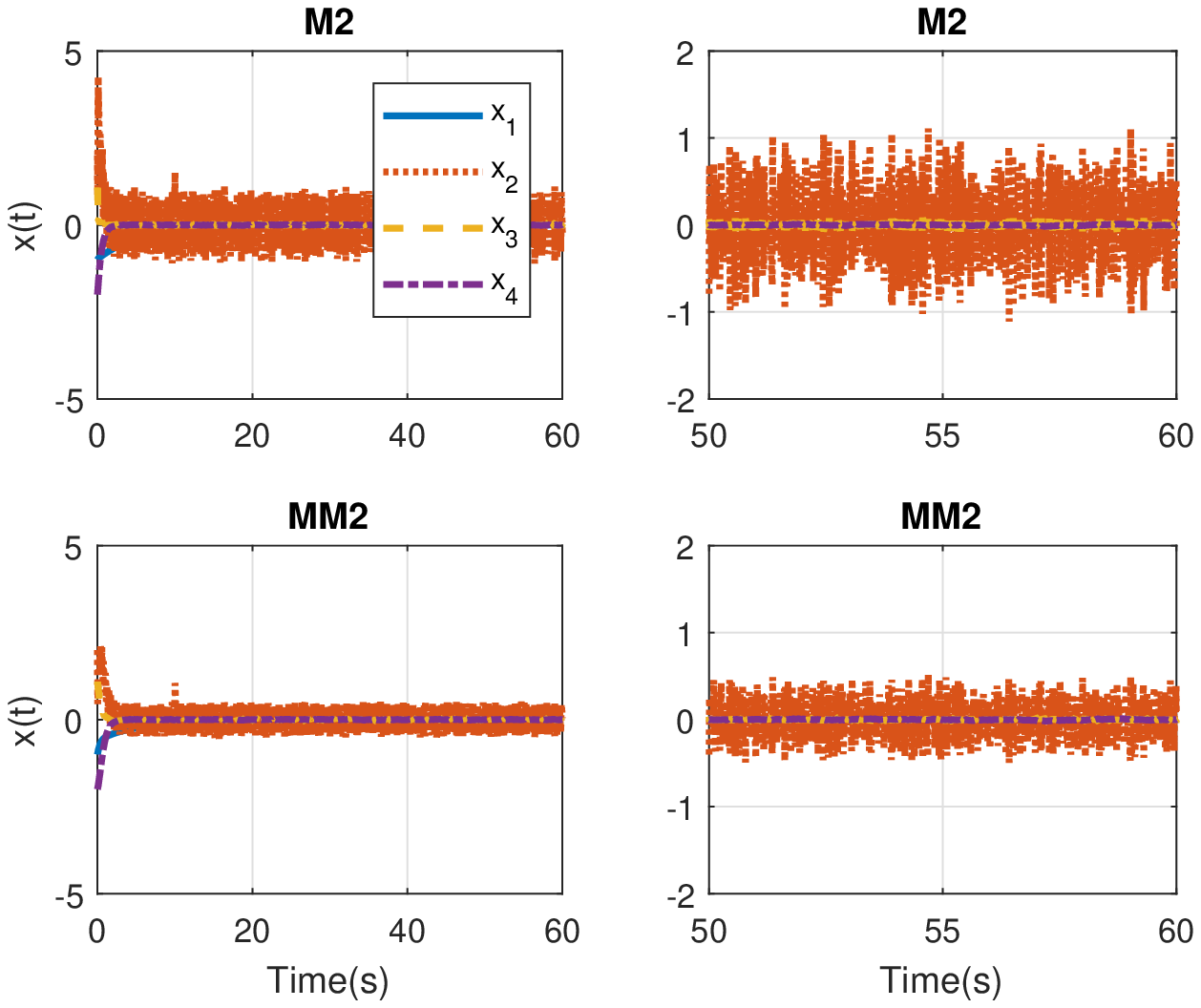}
      \caption{The evolution of the state variables using M2 and MM2 for the noisy case}
      \label{figx2n}
\end{figure}

\begin{figure}[!htb]
      \centering
      \includegraphics[height=2.5in]{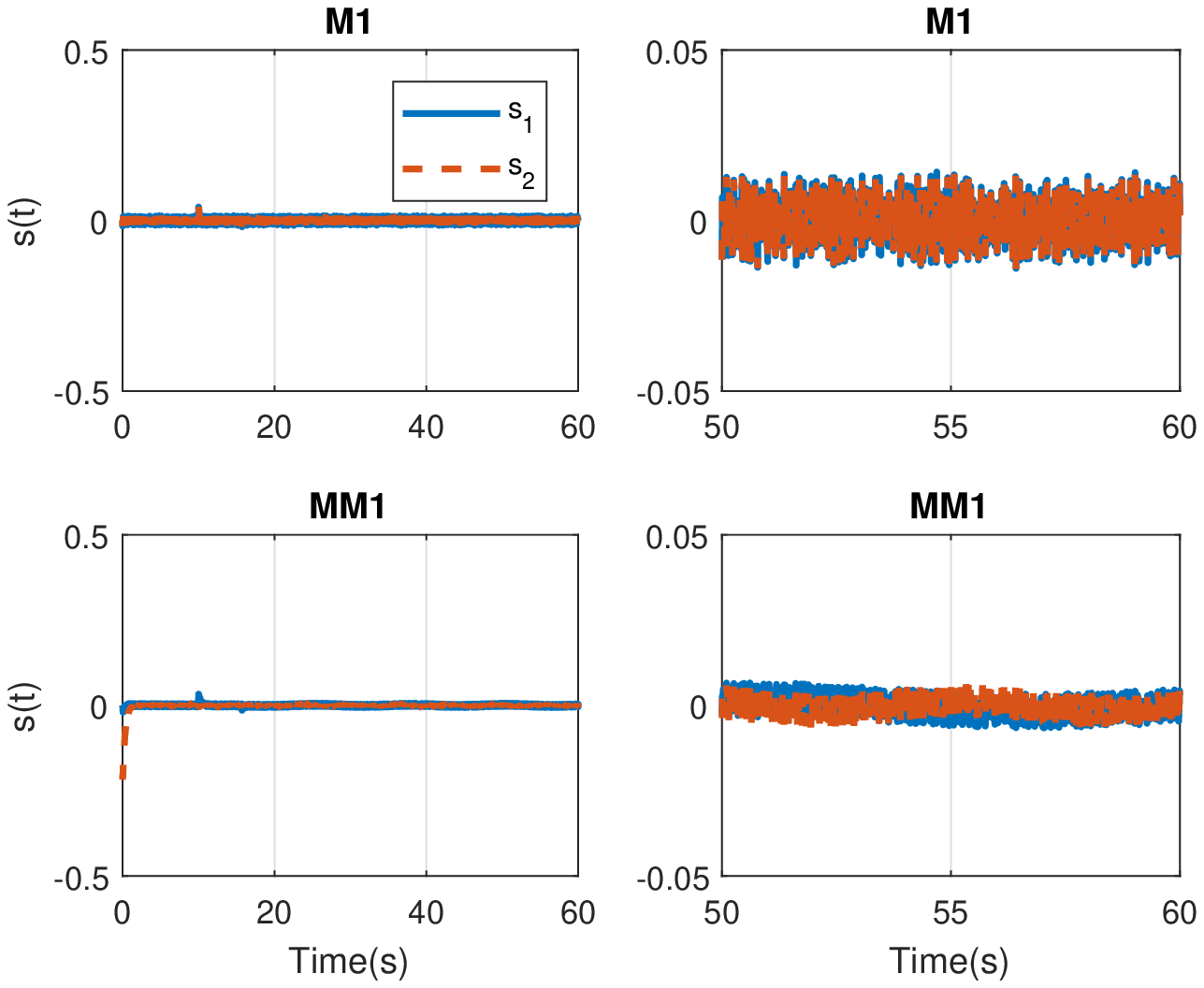}
      \caption{The evolution of the sliding functions using M1 and MM1 for the noisy case}
      \label{figs1n}
\end{figure}

\begin{figure}[!htb]
      \centering
      \includegraphics[height=2.5in]{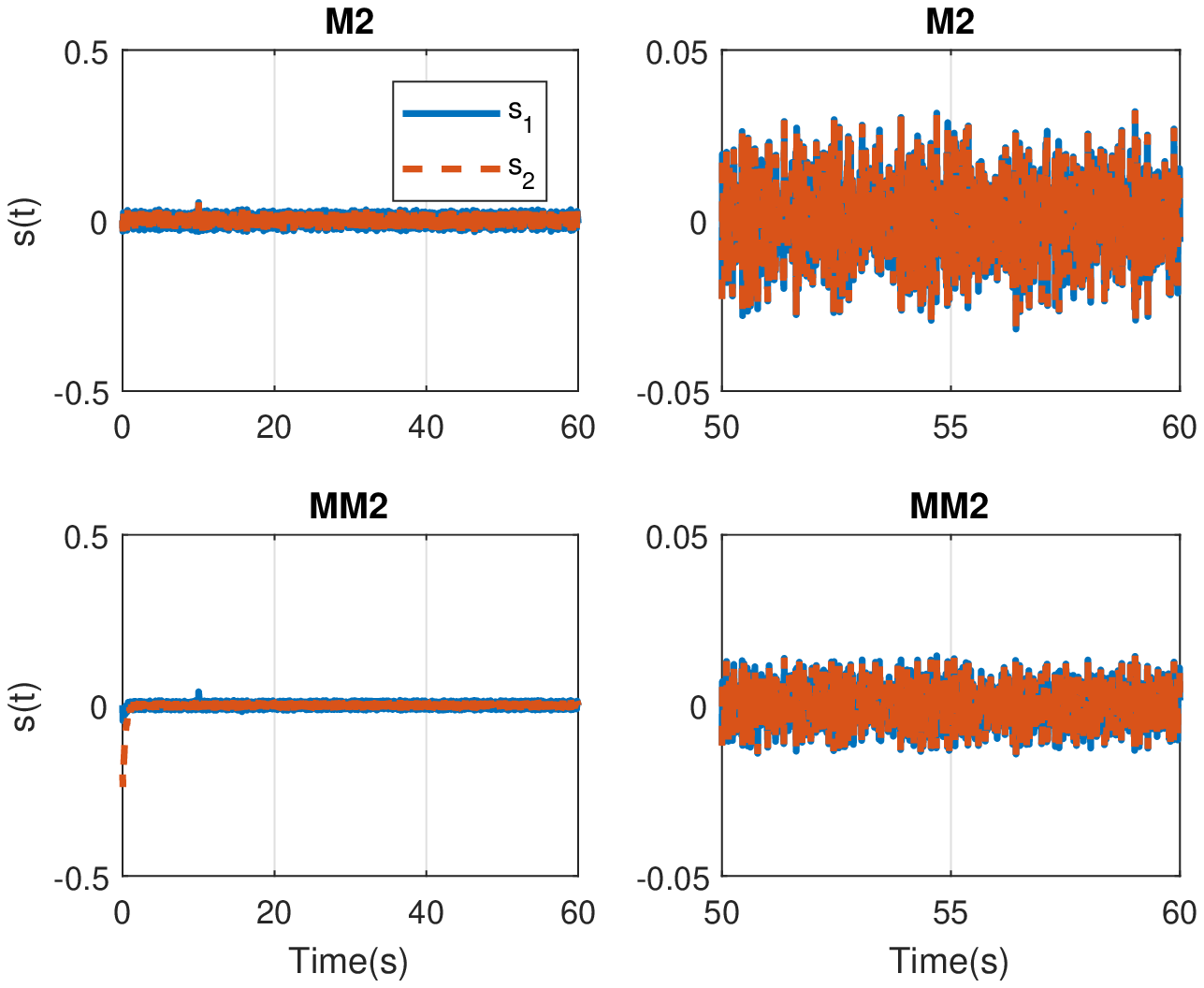}
      \caption{The evolution of the sliding functions using M2 and MM2 for the noisy case}
      \label{figs2n}
\end{figure}

\section{CONCLUSIONS}

The high gain nature of previous output feedback sliding mode control schemes was addressed, wherein a control signal of magnitude $O(1/T)$ could occur. New modifications were proposed to alleviate possible high gain control efforts, which can be of the order of $O(1)$. The theoretical analysis shows that the accuracy of the sliding mode using the modified versions of M1 and M2 are $O(T)$ and $O(T^2)$, while their original forms offer a $O(T^2)$ and $O(T^3)$ boundary layer for the sliding motion respectively. Simulation results have showed the effectiveness of the proposed schemes.

The proposed scheme is applied to linear sampled-data system with relative degree one. Future work will investigate control methods for systems with higher relative degree. Output feedback sliding mode control for nonlinear sampled-data system is also a possible future direction. Practical experiments will be conducted to verify the proposed approach.



\bibliographystyle{apalike-refs}
\bibliography{revised_SMoutput}

\end{document}